\documentclass[journal]{IEEEtran}
\usepackage{amsmath,amsfonts,amssymb,amsthm}
\usepackage[ruled,vlined,linesnumbered,noend]{algorithm2e}
\usepackage{booktabs}
\usepackage{multirow}
\usepackage{makecell}
\usepackage{capt-of}
\usepackage{array}
\usepackage[caption=false,font=normalsize,labelfont=sf,textfont=sf]{subfig}
\usepackage{textcomp}
\usepackage{stfloats}
\usepackage{url}
\usepackage{verbatim}
\usepackage{graphicx}
\usepackage{cite}
\usepackage{xcolor}
\usepackage[hidelinks]{hyperref}
\newtheorem{lemma}{Lemma}
\newtheorem{proposition}{Proposition}

\begin{document}

\title{ZOAF: Towards Efficient Zeroth-Order Optimization for Analog/RF Circuit
Design}

\author{Liyan~Tan,
        Yequan~Zhao,
        Jinming~Lu,
        Ben~F.~Jamroz,
        Ari~Feldman
        and~Zheng~Zhang% <-this % stops a space
\thanks{This work was supported by NIST under Award \#70NANB24H084}% <-this % stops a space
\thanks{L. Tan, Y. Zhao, J. Lu, and Z. Zhang are with the Department of Electrical and Computer Engineering, University of California, Santa Barbara, CA 93106 USA (e-mail: liyan\_tan@ucsb.edu; yequan\_zhao@ucsb.edu; jinminglu@ucsb.edu; zhengzhang@ece.ucsb.edu).}% <-this % stops a space
\thanks{B. F. Jamroz and A. Feldman are with the National Institute of Standards and Technology (NIST), Boulder, CO 80305 USA (e-mail: benjamin.jamroz@nist.gov; ari.feldman@nist.gov).}% <-this % stops a space% 
% \thanks{Manuscript received February XX, 2026; revised XXXX XX, 2026.}
}
% The paper headers
% \markboth{IEEE Transactions on Computer-Aided Design of Integrated Circuits and Systems,~Vol.~XX, No.~XX, Month~202X}%
% {Tan \MakeLowercase{\textit{et al.}}: ZOAF: Towards Efficient Zeroth-Order Optimization for Analog/RF Circuit Design}

% \IEEEpubid{0000--0000/00\$00.00~\copyright~2021 IEEE}
% Remember, if you use this you must call \IEEEpubidadjcol in the second
% column for its text to clear the IEEEpubid mark.

\maketitle

\begin{abstract}
Circuit optimization is an indispensable step in analog/RF IC design. Classical fast gradient-based optimization methods are typically infeasible due to lack of access to simulator source code and the technical barriers to implementing adjoint methods. Therefore, surrogate-based black-box optimization is widely used in practice; however, it can be costly to build and sensitive to hyperparameters, whereas population heuristics often suffer from slow convergence and large evaluation counts under tight simulator-call budgets. To address these limitations, we propose the Zeroth-Order Analog/RF Framework (ZOAF), which recovers gradient-descent directions from a small number of black-box circuit simulations, combining the benefits of both gradient-based optimization and black-box optimization. We also employ several surrogate-free techniques to improve the efficiency and accuracy, including (1) a hybrid ZO scheduling method that switches between random-direction ZO for budget-efficient exploration and coordinate-wise ZO for accurate late-stage refinement, (2) one-shot quasi-random multi-start to focus evaluations, and (3) a sliding-window monitor that triggers early stops and box-projected updates to maintain feasibility. Evaluated on three distinct schematics, ZOAF consistently outperforms state-of-the-art baselines, achieving the best median final value on every reported figure of merit---with up to an order-of-magnitude advantage in median peaking on the 22-parameter two-stage amplifier---together with the most robust worst-case behavior across seeds, while reducing simulator calls to convergence by $1.3$--$3.8\times$. Code is publicly available at \url{https://github.com/LiyanTan111/ZOAF}.

\end{abstract}


\section{Introduction}
\IEEEPARstart{A}{nalog/RF} circuit optimization is a core step in design closure: every tapeout depends on meeting targets for gain, bandwidth, phase margin, noise, and power across process, voltage, and temperature (PVT) corners and with layout parasitics~\cite{zhang2014stochastic,zhang2013stochastic}. In modern flows, accurate answers require SPICE/EM co-simulation, corner/Monte-Carlo sweeps, and layout-aware parasitic extraction---each call consumes minutes to hours---so any optimizer must spend its evaluation budget carefully \cite{daniel2004multiparameter,cui2018stochastic, li2005compact}. The search landscape is highly nonconvex with mixed continuous/discrete knobs, hard bounds, and occasional simulator non-convergence, and the key figures of merit often trade off against one another. These characteristics make the problem both expensive and brittle for general-purpose optimizers.

A long line of approaches has emerged. Model-based (intrusive) methods exploit analytical sensitivities or adjoint methods to enable gradient-descent updates. It is assumed that a designer can access and modify the codes of a simulator (e.g., HSPICE) to implement an adjoint method for fast gradient computation~\cite{wu1997circuit,georgieva2002feasible,nikolova2004adjoint,koziel2013reliable}.\footnote{Certain commercial software are identified in this paper to foster understanding. Such identification does not imply recommendation or endorsement by the National Institute of Standards and Technology, nor does it imply that the software identified are necessarily the best available for the purpose.} This has created a high barrier for practical deployment. Firstly, simulator codes are often unavailable to designers due to IP issues. Secondly, many designers may not have the required numerical analysis expertise to implement an adjoint-based optimization framework. Simulation-based stochastic methods---simulated annealing (SA)~\cite{gielen2002analog}, evolutionary algorithm (EA)~\cite{liu2009analog,barros2010analog}, genetic algorithm (GA)~\cite{zhou2022analog}, particle swarm optimization (PSO)~\cite{fakhfakh2010analog}, and covariance matrix adaptation evolution strategy (CMA-ES)~\cite{hansen2003reducing}---treat the circuit or system simulator as a black box and are robust to nonconvexity and mixed variables, yet they typically consume many evaluations and can struggle with constraints and noisy responses. Surrogate-assisted methods, most notably Gaussian-process (GP) Bayesian optimization (BO), improve sample efficiency by learning response surfaces and acquisition rules \cite{lyu2017efficient,lyu2018batch,zhang2019efficient}; however, kernel/hyperparameter selection, constrained and categorical design spaces, and scaling with dimension and batch parallelism introduce practical sensitivity. In realistic flows, reproducibility and manual tuning overhead also become nontrivial. Trust-region BO (e.g., TuRBO) was proposed to improve robustness in higher dimensional problems \cite{touloupas2021locomobo, eriksson2019scalable}.

These observations highlight four practical bottlenecks: (i) limited or no access to simulator gradients; (ii) tight evaluation budgets, further amplified by corners, Monte Carlo, and EM back-annotation~\cite{zhang2013uncertainty}; (iii) irregular, multi-modal, and sometimes non-smooth objectives caused by operating-region changes and solver artifacts; and (iv) hard constraints and box bounds that must always be satisfied.
In parallel, machine-learning-based methods (especially deep reinforcement learning and learned surrogates) have demonstrated strong potential for automated analog sizing; however, their practical effectiveness often depends on substantial offline training, reward/constraint engineering, and distributional alignment between training and deployment tasks.
Such assumptions are reasonable in repeated-design regimes, but can be restrictive in cold-start scenarios where a new circuit must be optimized quickly under a strict simulator-call budget and limited manual tuning.
A natural response is to treat the simulator strictly as a black box while still extracting useful search directions.
Zeroth-order (ZO) optimization takes exactly this stance: it estimates descent directions from a few perturbed function values, requires no gradients or adjoints, integrates cleanly with box projections, and can be engineered to be both budget-aware and noise-tolerant~\cite{liu2020primer,chen2017zoo}.
ZO further avoids simulator hooks, generalizes across topologies and device models, and lets us balance exploration and refinement through the choice of perturbation directions~\cite{zhao2023tensor, liu2020primer}.
ZO has been extensively studied in signal processing and machine learning, including black-box adversarial attacks~\cite{chen2017zoo, chen2019zo}, on-device training~\cite{zhao2025poor}, and large language model (LLM) fine-tuning~\cite{zhang2024revisiting, malladi2023fine, yang2024adazeta}. Nevertheless, to the best of our knowledge, ZO optimization has not yet been employed as an optimization primitive in analog/RF EDA flows.

Motivated by these challenges, we propose \textbf{ZOAF}, a multi-start projected zeroth-order (ZO) framework tailored to analog/RF sizing.
Our contributions are threefold.
First, we design a hybrid ZO schedule that uses ZO-RGE (\emph{random-direction gradient estimation}) for budget-efficient global exploration and switches to ZO-CGE (\emph{coordinate-wise gradient estimation}) for accurate local refinement.
Second, we introduce a one-shot quasi-random multi-start scheme with rank-based restarts that reuses a shared initialization pool and allocates simulator calls to promising regions.
Third, we develop sliding-window adaptation with constraint-aware acceptance, which tunes step and perturbation radii, triggers early restarts when progress stalls, and enforces box feasibility via clipped, greedily accepted updates.

Across three representative benchmarks, ZOAF consistently delivers a stronger accuracy–efficiency trade-off than GP-based BO and evolutionary baselines.
On the RF matching circuit, ZOAF attains the deepest return loss under the same evaluation budget, improving the best GP and surrogate-assisted EA competitors by over $10$\,dB.
On the op-amp and filter tasks, it matches or surpasses the best methods in terms of DC gain, gain--bandwidth product (GBP), and peaking/ripple/overshoot, while converging in $1.3$--$3.8\times$ fewer simulator calls and up to orders-of-magnitude less runtime.
These results indicate that the proposed coarse-to-fine projected zeroth-order schedule reliably converts tight evaluation budgets into high-quality analog/RF designs.
Importantly, the three design choices above align with the empirical findings: hybrid ZO scheduling improves early-to-late convergence, rank-based multi-start improves robustness to local traps, and sliding-window adaptation improves budget utilization and final solution quality.

\section{Related Work}\label{sec:relatedwork}

\subsection{Classical Bayesian Optimization (BO)}

Bayesian optimization (BO) is a sequential framework for expensive black-box objectives, including analog/analog-RF circuit sizing~\cite{rasmussen2003gaussian,lyu2017efficient}. It fits a probabilistic surrogate of $f$ (typically a Gaussian process) on observations $\mathcal{D}_N=\{(x_i,f(x_i))\}_{i=1}^N$, then maximizes an acquisition function to select the next design $x_{N+1}$~\cite{lyu2018batch,yin2024ado}. With a GP prior $f(\cdot)\sim\mathcal{GP}(m(\cdot),k_\theta(\cdot,\cdot))$ and noisy observations $y_i=f(x_i)+\varepsilon_i$, the posterior mean $\mu_N(x)$ and uncertainty $s_N(x)$ drive candidate selection. A standard acquisition is expected improvement (EI):
\begin{equation}
\label{eq:ei}
\mathrm{EI}(x)
= \mathbb{E}\bigl[\max\{f(x)-f^\star,0\}\,\big|\,\mathcal{D}_N\bigr],
\end{equation}
where $f^\star$ is the best observed value.

Other common acquisitions include probability of improvement (PI) and lower confidence bound (LCB).
Beyond the basic GP--EI loop, BO variants address practical circuit-design needs. Constrained and multi-objective BO handle specification limits and conflicting metrics~\cite{lyu2018multi}. Sparse/approximate GP and multi-fidelity BO improve scalability by reducing surrogate-update cost and combining fast/slow models~\cite{liu2013gaussian,zhang2014stochastic,zhang2019efficient,lyu2018multi}. For higher-dimensional parameter spaces, local trust-region BO (e.g., TuRBO) restricts search to adaptive local boxes and improves robustness~\cite{eriksson2019scalable,touloupas2021locomobo}.

\subsection{Classical Evolutionary Algorithms (EA)} Evolutionary algorithms (EA) are population-based, gradient-free methods that iteratively sample, evaluate, and select candidate designs. For continuous analog and analog/RF sizing, a common representative is evolution strategies (ES), which maintain a search center $m_t$ and step size $\sigma_t$, sample offspring around $m_t$, and recombine the best individuals. A simple $(\mu,\lambda)$–ES can be written as 

\vspace{-10pt}
\begin{subequations}\label{eq:es_sample_recombine}
\begin{align}
x_k &= m_t + \sigma_t z_k,\quad z_k\sim\mathcal{N}(0,I),\quad k=1,\ldots,\lambda, \label{eq:es_sample_recombine_x}\\
m_{t+1} &= \sum_{i=1}^{\mu} w_i\, x_{i:\lambda}. \label{eq:es_sample_recombine_m}
\end{align}
\end{subequations}
where $x_{i:\lambda}$ denotes the $i$-th best offspring by objective value and $w_i>0$ are recombination weights with $\sum_i w_i=1$. Modern ES variants such as CMA-ES adapt mutation covariance to local geometry and are widely used baselines for difficult continuous optimization~\cite{hansen2003reducing}. In analog/RF sizing, ES and related real-coded EAs (GA, differential evolution (DE), and PSO) are widely applied to optimize transistor- and component-level parameters under SPICE/ADS constraints~\cite{liu2009analog,barros2010analog,fakhfakh2010analog,zhou2022analog,nicosia2007evolutionary}. To reduce simulation cost, surrogate-assisted EAs combine evolutionary search with regression surrogates, including GP-assisted expensive optimization and circuit-specialized frameworks such as GASPAD~\cite{liu2013gaussian,liu2014gaspad,liu2013efficient,akinsolu2019parallel}. These methods provide an effective evolutionary alternative to pure BO pipelines for automated analog/RF optimization.

\subsection{Machine-Learning-Based Approaches (ML)}
Machine-learning-based methods are another major route for automated analog/analog-RF sizing, especially when repeated tasks justify up-front training cost. One line of work formulates sizing as a sequential decision process and learns policies from simulator feedback, including deep-RL frameworks such as \emph{L2DC}~\cite{wang2018learning}, \emph{AutoCkt}~\cite{settaluri2020autockt}, and actor--critic pipelines such as \emph{DNN-Opt}~\cite{budak2021dnn}. Later variants further improve robustness using circuit-aware representations, e.g., electrical-state features with sparse rewards~\cite{Uhlmann_2022} and circuit-attention mechanisms for variation-aware sizing~\cite{li2021circuit}.

In parallel, learning-augmented surrogates reduce expensive simulator queries by replacing part of exploration with predictions. Parasitic-aware methods use graph neural networks to estimate post-layout effects and combine uncertainty with BO-style acquisition for better sample efficiency~\cite{liu2021parasitic,gao2024post}. Related graph-based approaches such as \emph{GCX} build semi-supervised surrogates over circuit instances to guide search in constrained spaces~\cite{shahane2023graph}. Overall, ML-based methods are effective with sufficient data and careful reward/constraint design, while training-free black-box optimizers remain complementary in cold-start, tight-budget settings.

\subsection{Zeroth-Order (ZO) Optimization}
While BO, EA, and ML models have dominated black-box circuit sizing, they typically do not exploit the local geometric landscape (i.e., gradients) of the objective function. Zeroth-order (ZO) optimization bridges this gap by mimicking first-order gradient descent using only functional queries~\cite{flaxman2004online,nesterov2017random,liu2020primer}. Instead of relying on analytical sensitivities or intrusive adjoint methods, ZO methods estimate descent directions through finite-difference approximations along random or coordinate-wise perturbation vectors. By strictly controlling the perturbation radius and the sampling distribution, ZO estimators can recover mathematically unbiased gradients of a smoothed objective surrogate~\cite{chen2019zo,liu2020primer}.

Recently, ZO optimization has witnessed tremendous success in machine learning and signal processing, particularly in scenarios where backpropagation is infeasible or memory-prohibitive~\cite{liu2020primer}. Prominent applications include black-box adversarial attacks on deep neural networks~\cite{chen2017zoo}, resource-constrained on-device training~\cite{zhao2025poor,zhao2023tensor}, and memory-efficient fine-tuning of large language models~\cite{malladi2023fine,zhang2024revisiting,yang2024adazeta}. Despite its proven scalability and robustness to noisy, high-dimensional black-box functions, ZO optimization remains largely unexplored as a foundational optimization primitive in analog and RF EDA flows.

In the context of circuit sizing, ZO optimization offers a unique set of advantages: it bypasses the need to build and invert complex Gaussian process covariance matrices (a bottleneck in BO), avoids the sample-heavy heuristics of population-based EAs, and requires zero offline training. By explicitly approximating the local gradient, ZO methods can efficiently navigate the nonconvex ridges of circuit performance spaces, motivating the design of our proposed ZOAF framework.

\begin{figure*}[!t]
  \centering
  \includegraphics[width=\textwidth]{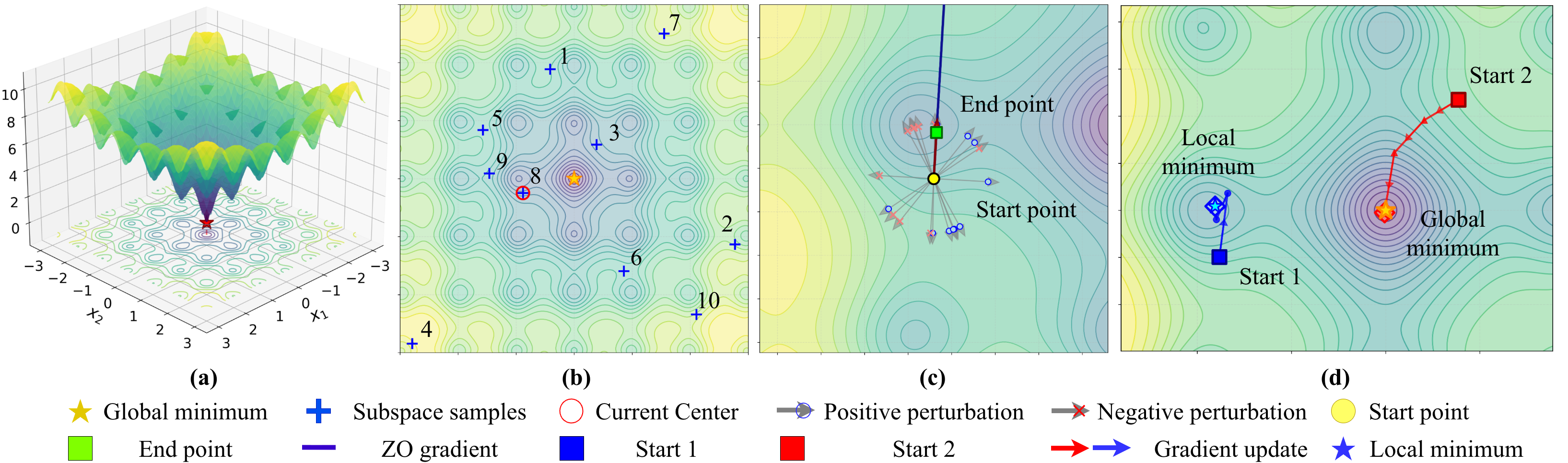}
  \vspace{-4mm}   
  \caption{Illustration of the proposed multi-start clipped ZO optimizer.
  \textbf{(a)} 2-D Ackley objective landscape.
  \textbf{(b)} One-shot initial sampling of candidate starting designs.
  \textbf{(c)} First stage: start from the top-ranked design and perform one ZO estimation.
  \textbf{(d)} Two example trajectories: the first start  gets trapped in a local minimum, the second start point reaches the global minimum.}
  \label{fig:optimization_diagram}
\end{figure*}

\section{Proposed Method}\label{sec:method}

We consider deterministic, single-objective black-box optimization for analog/RF circuit design. 
Let $x\in\mathbb{R}^d$ be the design vector (device sizes, biasing, $R$/$C$ values, etc.).
We focus on box-constrained problems:

\vspace{-10pt}

{\setlength{\abovedisplayskip}{2pt}%
 \setlength{\abovedisplayshortskip}{2pt}%
\begin{equation}
\label{eq:single_obj_box}
\max_{x\in\Omega} f(x), \qquad \Omega=[\ell,u]=\{x:\ \ell \le x \le u\}.
\end{equation}}

{\setlength{\textfloatsep}{2pt}
\setlength{\floatsep}{2pt}
\setlength{\intextsep}{2pt}
\setlength{\abovecaptionskip}{2pt}
\setlength{\belowcaptionskip}{0pt}
\begin{algorithm}[t]
\caption{ZOAF: Multi-start Projected Zeroth-Order Circuit Sizing (coarse RGE $\rightarrow$ fine CGE)}
\label{alg:mspzo}
\DontPrintSemicolon
\SetAlgoVlined
\SetKwInOut{Input}{Input}\SetKwInOut{Output}{Output}
\Input{$f$; bounds $[\ell,u]\subset\mathbb{R}^d$; budget $B$; base $(\eta_0,\mu_0)$; window $W$, tol $\delta$; factors $\gamma_\eta\!\in\!(0,1)$, $\gamma_\mu\!>\!1$; starts $M$; RGE directions $N$; init.\ dist.\ $p_{\text{init}}$.}
\Output{$x^\star$}
\BlankLine
\textbf{Init:} sample $\{x^{(m)}\}_{m=1}^{M}\!\sim\!p_{\text{init}}$; evaluate $f(x^{(m)})$; 
set $(x^\star,f^\star)\!\leftarrow\!\arg\max_m f(x^{(m)})$; 
$\mathrm{evals}\!\leftarrow\!M$; build pool $\mathcal{P}$ from $\{x^{(m)}\}$\;
\While{$\mathrm{evals}<B$}{
  select start $x$ from $\mathcal{P}$ via rank-based schedule\;
  $\eta\!\leftarrow\!\eta_0$; $\mu\!\leftarrow\!\mu_0$; phase $\leftarrow$ \textbf{RGE}; reset window stats\;
  \While{$\mathrm{evals}<B$}{
    \tcp{ZO gradient (Eqs.~\eqref{eq:zo_rge} and \eqref{eq:zo_cge})}
    \uIf{phase $=$ \textbf{RGE}}{
      compute $\hat g$ via RGE with $N$ directions; $\mathrm{evals} \mathrel{+}= 2N$\;
    }
    \Else{
      compute $\hat g$ via CGE over all coordinates; $\mathrm{evals} \mathrel{+}= 2d$\;
    }
    $x^{\mathrm{cand}}\!\leftarrow\!\operatorname{clip}(x+\eta\hat g,\,\ell,\,u)$\;
    $y^{\mathrm{cand}}\!\leftarrow\!f(x^{\mathrm{cand}})$; $\mathrm{evals}\!+\!+$\;
    \If{$y^{\mathrm{cand}}>f(x)$}{
      $x\!\leftarrow\!x^{\mathrm{cand}}$; 
      \If{$y^{\mathrm{cand}}>f^\star$}{$x^\star\!\leftarrow\!x$; $f^\star\!\leftarrow\!y^{\mathrm{cand}}$}
    }
    update window stats\;
    \If{$\overline{\Delta}_t\le\delta$}{
      $\eta\!\leftarrow\!\gamma_\eta\eta$; $\mu\!\leftarrow\!\gamma_\mu\mu$\;
    }
    \If{switch-condition \& budget-guard (see Sec.~\ref{subsec:window})}{
      phase $\leftarrow$ \textbf{CGE}\;
    }
    \If{early-convergence (see Sec.~\ref{subsec:window})}{
      \textbf{break}\;
    }
  }
}
\Return $x^\star$\;
\end{algorithm}
\vspace{-6pt}
}

Algorithm~\ref{alg:mspzo} is stated in maximization form, using the greedy-acceptance condition $y^{\mathrm{cand}}>f(x)$; for objectives that are minimized (ripple, peaking, $|S_{22}|$), we negate $f$ before passing it to the algorithm, so all objectives are handled uniformly.
The convergence analysis in the Appendix uses the equivalent minimization form $F(x)=-f(x)$.

To tackle this problem, we propose \textbf{ZOAF}, a two-phase, multi-start projected ZO optimizer for analog/RF sizing under a fixed simulation budget.
In this section, we detail the key components of ZOAF: Sec.~\ref{subsec:zo_hybrid} introduces the hybrid ZO direction scheme (ZO-RGE: random-direction gradient estimation, and ZO-CGE: coordinate-wise gradient estimation), Sec.~\ref{subsec:clip} explains the projected updates for handling box constraints, Sec.~\ref{subsec:init} presents the one-shot multi-start initialization, and Sec.~\ref{subsec:window} describes the sliding-window adaptation and restart logic; Algorithm~\ref{alg:mspzo} summarizes the full algorithm flow.
An intuitive overview of the multi-start workflow and optimization trajectories is shown in Fig.~\ref{fig:optimization_diagram}.

\subsection{ZO Gradient Estimation \& Hybrid Schedule}\label{subsec:zo_hybrid}
Zeroth-order (ZO) methods perform gradient-based optimization using only function evaluations~\cite{liu2020primer,chen2017zoo,flaxman2004online,nesterov2017random}. 
In our setting, each query runs a circuit simulator at parameter vector $x$ (device sizes, biases, $R$/$C$ values, etc.) to obtain the objective (e.g., gain, bandwidth, power).
At iteration $t$, we estimate a descent direction $\hat g_t$ from symmetric two-point probes and update $x_{t+1}=x_t-\eta_t\hat g_t$ with stepsize $\eta_t$ and a small perturbation radius $\mu>0$ that controls the finite-difference smoothing.

\noindent\textit{ZO-RGE: random-direction gradient descent} (cost $2N$ evaluations per update).
We draw $N$ i.i.d.\ random directions $\{u_{t,k}\}_{k=1}^N$ from a zero-mean isotropic distribution with
\[
\mathbb{E}[u_{t,k}]=0,\qquad \mathbb{E}[u_{t,k}u_{t,k}^\top]=I_d.
\]
In practice we use either Rademacher vectors with independent coordinates $u^{(i)}\in\{-1,+1\}$ (so $\mathbb{E}[u^{(i)}]=0$ and $\mathrm{Var}(u^{(i)})=1$), or unit-sphere directions obtained as $u = \sqrt{d}\,g/\|g\|$ with $g\sim\mathcal{N}(0,I_d)$, which also satisfy the above moment conditions.
Given these directions, we apply two-sided perturbations and average the resulting estimates~\cite{liu2019signsgd}:
\begin{equation}
\label{eq:zo_rge}
\hat g_t = \frac{1}{2\mu N}\sum_{k=1}^{N}
\bigl(f(x_t+\mu u_{t,k}) - f(x_t-\mu u_{t,k})\bigr)\,u_{t,k},
\end{equation}

Here $\mu$ is chosen as a small fraction of the box width (so that perturbations stay in a local neighborhood while averaging out small numerical irregularities), and $N$ trades off estimator variance against query cost: larger $N$ reduces randomness in $\hat g_t$ but costs $2N$ simulator calls per iteration, so in practice we keep $N$ in the low single digits for expensive analog/RF circuits.

\noindent\textit{ZO-CGE: coordinate gradient descent} (cost $2d$ evaluations per update).
Alternatively, we perturb each coordinate axis $e_i$ symmetrically to estimate all partial derivatives and then perform a single vector update:
\begin{equation}
\label{eq:zo_cge}
\hat g_t[i] = \frac{f(x_t+\mu e_i) - f(x_t-\mu e_i)}{2\mu}, \quad i=1,\ldots,d.
\end{equation}

This estimator uses the same perturbation radius $\mu$ but requires $2d$ simulator calls per step, yielding a low-variance, coordinate-wise ZO gradient that is well suited for late-stage local refinement when $d$ is moderate.

\noindent\textbf{Hybrid schedule: coarse RGE $\rightarrow$ fine CGE.}
RGE has low query cost and scales weakly with dimension, making it well suited for fast exploration, whereas CGE is more expensive but lower-variance and accurate for local refinement. ZOAF therefore runs RGE to reach a promising basin and, based on a sliding-window monitor, either switches to CGE when progress slows but has not stalled, or terminates the stage and restarts from the shared quasi-random pool when improvement collapses. A small budget guard prevents switching too late, and this coarse-to-fine handoff consistently improves early convergence and final quality under the same simulator budget~\cite{zhao2023tensor}.

\noindent\textbf{Scheduling rule.}
Each stage starts in RGE. We run a sliding-window test on recent improvements, gradient norms, and remaining budget. If the switch condition is met, the phase flips to CGE; if the early-convergence condition is met instead, the stage stops and a new start is drawn from the shared pool. In practice, the RGE$\rightarrow$CGE switch can also be user-specified or tied to a simple budget threshold, while reusing the same early-convergence rule.

\subsection{Theoretical Insight on Hybrid Scheduling}
The strategic transition from RGE to CGE within our hybrid schedule is rigorously motivated by the statistical properties of the underlying ZO estimators. The Appendix (Section~\ref{app:zo_rge_proofs}) formally establishes these characteristics for the two-point ZO-RGE estimator defined in \eqref{eq:rge_est}. Specifically, Proposition~\ref{prop:unbiased} proves that the RGE estimator is strictly unbiased with respect to the Gaussian-smoothed gradient $\nabla f_{\mu}(x)$. Furthermore, Proposition~\ref{prop:smooth_bias} bounds the deterministic smoothing bias by $\mathcal{O}(L\mu\sqrt{d})$, demonstrating that a carefully chosen perturbation radius $\mu$ effectively controls the deviation from the true gradient of an $L$-smooth objective. 

However, the critical limitation of RGE emerges from its variance. Proposition~\ref{prop:var} demonstrates that the mean-square error of the RGE estimator scales as $\mathcal{O}(d^2/N)$ for a $G$-Lipschitz function. This theoretical bound perfectly explains the operational dynamics of our framework: the RGE update in \eqref{eq:zo_rge} is highly effective for initial budget-efficient exploration---requiring only $2N \ll 2d$ queries to find a viable descent direction---but becomes severely variance-limited in higher-dimensional circuit sizing (e.g., $d \approx 20$). As the search approaches a local optimum and the true gradient magnitude shrinks, the $\mathcal{O}(d^2/N)$ sampling noise drowns out the gradient signal, stalling convergence.

Conversely, the ZO-CGE estimator in \eqref{eq:zo_cge} circumvents this directional sampling variance entirely by systematically probing all $d$ orthogonal coordinate bases. In light of the variance bounds established in Proposition~\ref{prop:var}, transitioning to CGE incurs a fixed cost of $2d$ evaluations per step but recovers the deterministic, coordinate-wise fidelity necessary for late-stage convergence. This precision is indispensable for resolving fine-grained performance trade-offs, such as deepening the return loss $|S_{22}|$ in RF matching networks or suppressing passband ripple in cascaded amplifiers. Consequently, our hybrid schedule deploys RGE for rapid, budget-efficient basin discovery, then dynamically switches to CGE for high-precision local refinement, thereby maximizing the utility of every expensive simulator call.

\subsection{Constraint Handling \& Acceptance} \label{subsec:clip}

Given the ZO estimators above, we enforce box constraints $[\ell,u]$ by clipping both probes and updates. 
This avoids infeasible simulator calls and keeps all evaluations within bounds.

\paragraph{Clipped probes (shared).}
For any probe direction $v$ (a random $u_{t,k}$ in RGE or a basis vector $e_i$ in CGE), we form symmetric two-point queries inside the box:
\begin{equation}\label{eq:clip-probes}
x^{\pm} \;=\; \operatorname{clip}\bigl(x_t \pm \mu_t v,\, \ell,\, u\bigr),
\qquad
v \in \{u_{t,k}\}\ \text{or}\ v=e_i.
\end{equation}

\paragraph{Projected greedy update.}
After obtaining $\hat g_t$ from RGE/CGE, we take a clipped ascent step and accept it only if it improves the objective:
\begin{subequations}\label{eq:clip-update}
\begin{align}
x_{\mathrm{cand}} &= \operatorname{clip}\bigl(x_t + \eta_t \hat g_t,\, \ell,\, u\bigr), \label{eq:clip-update-cand}\\
x_{t+1} &=
\begin{cases}
x_{\mathrm{cand}}, & f(x_{\mathrm{cand}}) > f(x_t),\\
x_t, & \text{otherwise},
\end{cases} \label{eq:clip-update-next}\\
f^\star_{t+1} &= \max\{f^\star_t,\, f(x_{t+1})\}. \label{eq:clip-update-best}
\end{align}
\end{subequations}
This greedy rule concentrates evaluations on improving designs while guaranteeing feasibility by construction.

\subsection{Initial Sampling \& Multi-start}
\label{subsec:init}
ZO can start from any seed, but a single run often falls into a local minimum and converges quickly to that basin. Under tight simulator-call budgets, using several short starts from different regions raises the chance of reaching a better (near-global) solution.

We therefore build a one-shot pool at the start and reuse it for all restarts. We draw $M$ points $x^{(m)}\!\sim\!p_{\text{init}}$ on $[\ell,u]$ (Sobol/LHS/uniform, or a simple combination), evaluate $f(x^{(m)})$, and store
\begin{equation*}
\mathcal{P}=\{(x^{(m)},\,f(x^{(m)}))\}_{m=1}^{M}.
\end{equation*}
No new samples are created during restarts.
The $M$ initial evaluations are counted toward the total simulator-call budget $B$, so the pool construction cost is fully reflected in all reported evaluation counts.

At stage boundaries, we rank all pool points by their objective value $f(x^{(m)})$ (higher is better) and pick the next start by a rank-based rule: early stages prefer top-ranked designs (exploitation), and later stages gradually allow lower ranks (exploration).

\subsection{Sliding Window \& Early Convergence}\label{subsec:window}
To save simulator calls and avoid lingering in a poor basin, we track short-term progress with a sliding window and stop a stage when progress clearly stalls.

We maintain statistics over the last $W$ steps (e.g., $W{=}8$). Let
\begin{subequations}
\begin{align}
\mathcal{I}_t &= \{\Delta_j\}_{j=t-W+1}^{t}\!,\ \ \Delta_j=f^\star_j-f^\star_{j-1}, \label{eq:window_I}\\
\mathcal{G}_t &= \{\|\hat g_j\|_2\}_{j=t-W+1}^{t}. \label{eq:window_G}
\end{align}
\end{subequations}
and define
\begin{subequations}
\begin{align}
\overline{\Delta}_t &= \mathrm{mean}(\mathcal{I}_t), \label{eq:window_dbar}\\
s_{\Delta,t} &= \mathrm{std}(\mathcal{I}_t), \label{eq:window_s}\\
\bar{g}_t &= \mathrm{mean}(\mathcal{G}_t). \label{eq:window_gbar}
\end{align}
\end{subequations}
with $m_t$ the least-squares slope of $\{(j,\Delta_j)\}$ over the window.

\emph{Early convergence} is declared when at least two of the following hold:
\begin{subequations}
\begin{align}
\overline{\Delta}_t &\le \tau_{\Delta}, \label{eq:window_cond1}\\
\bar{g}_t &\le \tau_{g}, \label{eq:window_cond2}\\
\lvert m_t \rvert &\le \tau_{m}, \label{eq:window_cond3}\\
s_{\Delta,t} &\le \tau_{s}. \label{eq:window_cond4}
\end{align}
\end{subequations}
We fix $W{=}8$ and use the same thresholds for all circuits:
$\tau_{\Delta}=10^{-4}$ (avg\_improvement),
$\tau_{g}=10^{-2}$ (avg\_grad\_norm),
$\tau_{m}=10^{-5}$ (abs\_trend\_slope),
and $\tau_{s}=10^{-6}$ (std\_improvement).

When triggered, the stage terminates and a restart is launched from the shared pool $\mathcal{P}$ (Sec.~\ref{subsec:init}).

If recent improvement is small but not yet stagnant (e.g., $\overline{\Delta}_t \le \delta$ only), we apply the soft-stalling monitor in Algorithm~\ref{alg:mspzo}: $\overline{\Delta}_t$ is computed as the mean of the recent window $\mathcal{I}_t=\{\Delta_j\}_{j=t-W+1}^{t}$ (Eq.~\eqref{eq:window_I}) and compared with $\delta$; if triggered, the step size and perturbation radius are updated as $\eta\leftarrow\gamma_\eta\eta$ and $\mu\leftarrow\gamma_\mu\mu$. This reduces overshooting and broadens probing to recover a usable signal without wasting budget in the current basin.

\vspace{-6pt}

\section{Experiments}\label{sec:experiments}

\noindent\textbf{Experimental setup.}
We evaluate ZOAF on three analog/RF circuits:
(i) an RF matching network (Fig.~\ref{fig:rf_match}),
(ii) a three-stage operational amplifier (Fig.~\ref{fig:opamp3}(a)), and
(iii) a two-stage cascaded signal-conditioning amplifier with 22 design variables (Fig.~\ref{fig:opamp3}(b)).
The RF case is simulated in ADS~2021 via netlist edits, while the amplifiers use PySpice~1.5.
All methods share the same simulator-call budget---including each method's initialization phase (e.g., the initial design-of-experiments for GP-based methods and the initial population for EA-based methods)---and run on a 32-core Xeon workstation.
For all stochastic methods, hyperparameters follow their public implementations. Statistical reporting differs by benchmark: on the RF matching and three-stage op-amp tasks we report 10-run statistics (best, mean $\pm$ standard deviation), whereas on the 22-parameter two-stage cascaded amplifier we report 100-seed distributional statistics (min, 10th percentile, median, mean, 90th percentile, max) so that the larger seed sample can resolve typical-case versus adversarial-seed behavior.
We benchmark ZOAF against representative Bayesian and evolutionary baselines, including GP-based Bayesian optimization (EI, PI, LCB, and Mixer-AF), DE, GASPAD, TuRBO-1, CMA-ES, and PSO.
Gradient-based methods (e.g., adjoint-based optimizers) are not included as baselines because they require access to simulator source code to implement sensitivity computation---an assumption that is unavailable with commercial tools such as ADS and PySpice and that constitutes the primary motivation for the black-box setting adopted in this work.
Results are reported in terms of total simulator calls, convergence speed, and achieved figure of merit (FOM).
The ``Conv.'' column in all tables counts the simulator call at which each method first reaches its reported best FOM; one epoch corresponds to exactly one simulator call for every method, so the comparison is strictly fair.
These three benchmarks were selected because their FOMs are real-analytic functions of the passive component values within the feasible box---guaranteeing that the $L$-smoothness assumption (Assumption~2) holds with a well-controlled Lipschitz constant and that the objectives are bounded (Assumption~5), which directly tightens all three terms in the convergence bound~\eqref{eq:convergence_bound}.
Their moderate dimensions ($d\le 22$) place them in a regime where RGE's linear-in-$d$ query cost enables efficient early exploration and CGE's deterministic finite-difference sweep remains tractable, matching the algorithmic intent of the hybrid schedule.

% ===================== 4.1 RF MATCHING =====================
\subsection{RF Matching Network}

\noindent We first examine a passive RF matching network (Fig.~\ref{fig:rf_match}) terminated with 50\,$\Omega$ loads on a fixed microstrip substrate. 
Fifteen lumped elements are tuned, and the objective is to minimize the return loss $|S_{22}|$ at the target frequency of 94\,GHz. 
The circuit is treated as a black box, and all methods share the same fixed simulator-call budget for fair comparison.

\noindent\textbf{Result comparison.}
Under a 150-call budget and identical initialization (Table~\ref{tab:rf_s22}), ZOAF achieves the deepest single-run return loss of \(|S_{22}| = -66.46\)\,dB.
We note that improvements below $-40$\,dB carry diminishing engineering returns: measurement uncertainties, fabrication tolerances, and the near-unity power-transmission factor ($1-|S_{22}|^2\approx 1$ below $-30$\,dB) collectively make this region practically saturated.
The more meaningful comparison therefore lies in \emph{reliability across runs}.
In the Mean\,$\pm$\,Std column, GP-EI, GP-PI, and GP-LCB average only $-22.52$, $-26.44$, and $-32.67$\,dB respectively---failing to reach the practical $-30$\,dB threshold on average---while GP-Mixer-AF barely clears $-40$\,dB at $-40.11$\,dB mean.
ZOAF, by contrast, reaches $-57.69 \pm 6.37$\,dB on average, reliably exceeding the practical matching threshold in every independent run with a tighter spread than GP-Mixer-AF ($\pm 6.68$\,dB).
These results demonstrate that ZOAF’s coarse-to-fine schedule consistently converts the 150-call budget into well-matched designs on this 15-parameter task, where most baselines fail to do so reliably.

\begin{figure*}[t]
\centering
\begin{minipage}[c]{0.48\textwidth}
\centering
\begingroup
\setlength{\tabcolsep}{2pt}
\renewcommand{\arraystretch}{0.97}
\setlength{\belowcaptionskip}{6pt}
\footnotesize
\captionof{table}{$|S_{22}|$ (dB) Comparison on the RF Matching Network.}
\label{tab:rf_s22}
\begin{tabular}{lcccc}
\toprule
\textbf{Category} & \textbf{Method} & \textbf{Evals} & \makecell{\textbf{Best $S_{22}$}\\\textbf{(dB)}} & \makecell{\textbf{Mean $\pm$ Std of}\\\textbf{$S_{22}$ (dB)}} \\
\midrule
\multirow{4}{*}{GP}
 & GP-EI       & 150 & -25.47 & $-22.52 \pm 2.31$ \\
 & GP-PI       & 150 & -30.34 & $-26.44 \pm 3.38$ \\
 & GP-LCB      & 150 & -36.69 & $-32.67 \pm 2.89$ \\
 & GP-Mixer-AF & 150 & -48.20 & $-40.11 \pm 6.68$ \\
\midrule
\multirow{2}{*}{EA}
 & DE          & 150 & -33.42 & $-30.29 \pm 2.60$ \\
 & GASPAD      & 150 & -50.31 & $-45.75 \pm 3.12$ \\
\midrule
\multirow{1}{*}{ZO}
 & \textbf{ZOAF} & 150 & \textbf{-66.46} & $\mathbf{-57.69 \pm 6.37}$ \\
\bottomrule
\end{tabular}
\endgroup
\end{minipage}
\hfill
\begin{minipage}[c]{0.48\textwidth}
\centering
\vspace{0pt}
\includegraphics[width=0.9\textwidth]{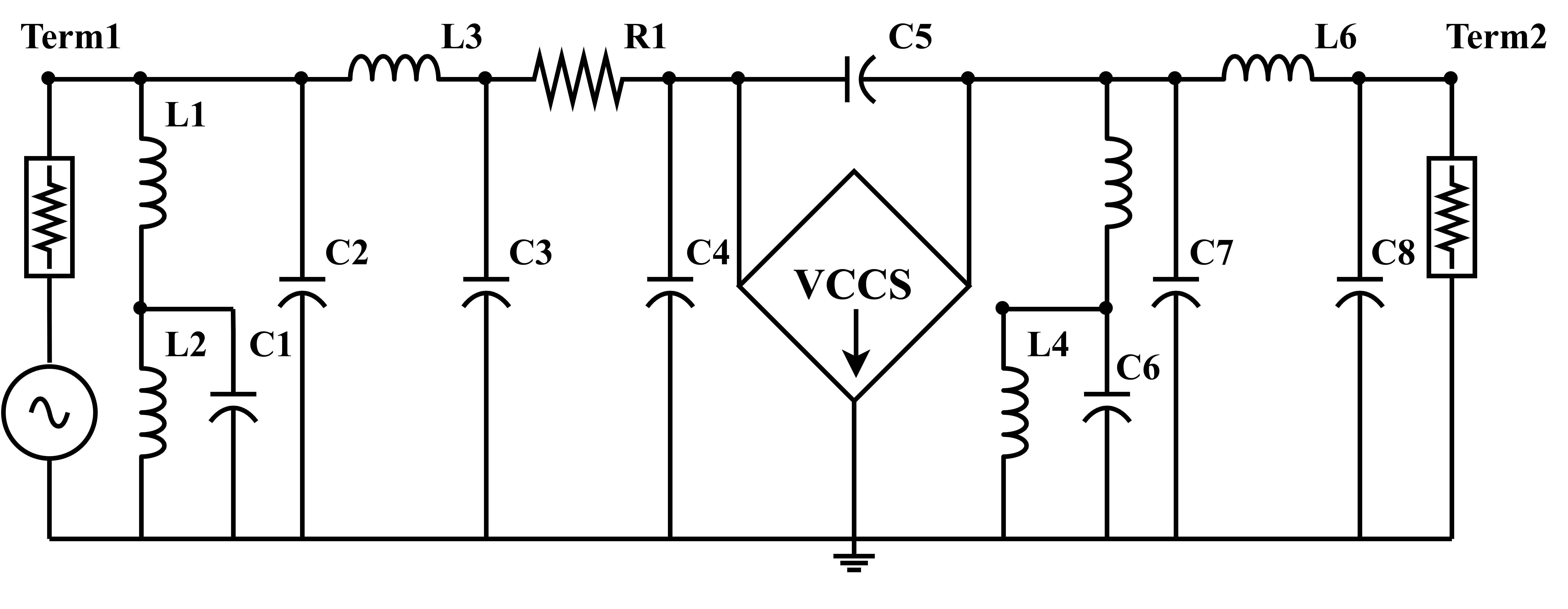}
\vspace{-1mm}
\captionof{figure}{Schematics for RF matching network.}
\label{fig:rf_match}
\end{minipage}
\end{figure*}

\begin{figure*}[!t]
  \centering
  \includegraphics[width=\textwidth]{Figs/circuit2_3.png}  
  \caption{Schematics for (a) Three-stage op-amp. (b) Two-stage cascaded signal-conditioning amplifier.} 
  \label{fig:opamp3}
\end{figure*}

% \vspace{-6pt}

% ===================== 4.2 THREE-STAGE OP-AMP =====================

\subsection{Three-Stage Operational Amplifier}

\begin{table*}[t]
\centering
\begingroup
\footnotesize
\setlength{\belowcaptionskip}{6pt}
\caption{Three-Stage Op-Amp: Best FOMs With Runtime, Convergence Epoch (Conv.), and Evaluation Counts (Evals) Under a 300-Call Simulator Budget.}
\label{tab:opamp3_results}
\begin{tabular}{lcccc cccc cccc}
\toprule
\multirow{2}{*}{\textbf{Method}} 
  & \multicolumn{4}{c}{\textbf{Gain (–)}} 
  & \multicolumn{4}{c}{\textbf{GBP (MHz)}} 
  & \multicolumn{4}{c}{\textbf{Power efficiency (1/W)}} \\
\cmidrule(lr){2-5} \cmidrule(lr){6-9} \cmidrule(lr){10-13}
 & \textbf{Value} & \textbf{Time (s)} & \textbf{Conv.} & \textbf{Evals}
 & \textbf{Value} & \textbf{Time (s)} & \textbf{Conv.} & \textbf{Evals}
 & \textbf{Value} & \textbf{Time (s)} & \textbf{Conv.} & \textbf{Evals} \\
\midrule
\textbf{ZOAF}
  & \textbf{$9.73\times 10^{8}$}  & 1.1    & 72  & 290
  & \textbf{$1.01\times 10^{7}$}  & 2.9    & 68  & 290
  & $9.33\times 10^{2}$           & 0.7    & 290  & 290 \\
CMA-ES
  & \textbf{$9.73\times 10^{8}$}  & 1.0    & 200 & 300
  & $5.66\times 10^{4}$           & 2.6    & 280 & 300
  & $9.05\times 10^{2}$           & 0.8    & 140 & 300 \\
TuRBO-1
  & $1.48\times 10^{8}$           & 1385.9 & 165 & 300
  & $1.38\times 10^{4}$           & 1389.3 & 40  & 300
  & \textbf{$9.40\times 10^{2}$}  & 1186.8 & 290 & 300 \\
GP-EI
  & $6.29\times 10^{4}$           & 2326.2 & 100 & 300
  & $1.36\times 10^{4}$           & 2250.0 & 105 & 300
  & $8.27\times 10^{2}$          & 2029.7 & 65  & 300 \\
GP-PI
  & $8.62\times 10^{5}$           & 2261.7 & 85  & 300
  & $1.98\times 10^{2}$           & 2624.7 & 120 & 300
  & $8.04\times 10^{2}$           & 2115.2 & 180 & 300 \\
AutoCkt
  & \textbf{$2.38\times 10^{7}$}  & 2.4    & 150 & 300
  & $7.68\times 10^{5}$           & 5.9    & 300 & 300
  & $8.25\times 10^{2}$           & 2.1    & 250 & 300 \\
DNN-Opt
  & \textbf{$5.21\times 10^{7}$}  & 2.5    & 300 & 300
  & $5.42\times 10^{6}$           & 6.1    & 300 & 300
  & $8.84\times 10^{2}$           & 2.4    & 300 & 300 \\
\bottomrule
\end{tabular}
\endgroup
\end{table*}

\noindent We next evaluate a cascaded three-stage non-inverting amplifier (Fig.~\ref{fig:opamp3}(a)). 
The decision vector contains ten components,
$R_{1\ldots 8}\!\in[0.1,100]$\,k$\Omega$ and $C_{1,2}\!\in[0.1,100]$\,nF, 
and each op-amp is modeled as a high-gain macromodel.
All methods optimize the same single-objective FOMs under a fixed simulator-call budget.

\noindent\textbf{Optimization setup.} Let $A_{\mathrm{DC}i}$ denote the DC gain of the $i$-th stage and $f_{3\mathrm{dB}}(x)$ the closed-loop $-3$\,dB bandwidth. We consider three single-objective FOMs---DC gain, gain--bandwidth product, and power efficiency---all maximized over $x$:
\begingroup
\setlength{\abovedisplayskip}{2pt}
\setlength{\belowdisplayskip}{4pt}
\setlength{\abovedisplayshortskip}{2pt}
\setlength{\belowdisplayshortskip}{4pt}
\begin{subequations}\label{eq:obj_3stage}
\begin{align}
A_{\mathrm{DC}}(x)    &= A_{\mathrm{DC1}}\,A_{\mathrm{DC2}}\,A_{\mathrm{DC3}}, \label{eq:obj_gain}\\
\mathrm{GBP}(x)       &= A_{\mathrm{DC}}(x)\,f_{3\mathrm{dB}}(x), \label{eq:obj_gbw}\\
P_{\mathrm{total}}(x) &= V_{\mathrm{CC}}\!\left(I_{q1}+I_{q2}+I_{q3}\right)+3\,P_{\mathrm{opamp}}, \label{eq:obj_eff_ptotal}\\
\eta(x)               &= \frac{A_{\mathrm{DC}}(x)}{P_{\mathrm{total}}(x)}. \label{eq:obj_eff}
\end{align}
\end{subequations}
\endgroup

% \vspace{-6pt}

\noindent\textbf{Result comparison.}
For this 300-call setting, all methods use the same target simulator-call budget while preserving their native update granularity; therefore, minor differences in realized calls (e.g., 290 vs. 300) can occur and are reported explicitly in Table~\ref{tab:opamp3_results}.
For clarity, bold entries in Table~\ref{tab:opamp3_results} indicate the best value in each metric column (higher is better for gain, GBP, and power efficiency).
Table~\ref{tab:opamp3_results} shows that ZOAF matches the best DC gain ($9.73 \times 10^{8}$) while requiring far fewer epochs (72 vs.\ 200 for CMA\mbox{-}ES) at comparable runtime.
For GBP, ZOAF attains $1.01 \times 10^{7}$\,MHz, corresponding to approximately $180\times$ over CMA\mbox{-}ES ($5.66 \times 10^{4}$\,MHz) and $740\times$ over GP\mbox{-}EI ($1.36 \times 10^{4}$\,MHz), while also converging sooner (68 vs.\ 280 epochs for CMA\mbox{-}ES).
On power efficiency, under the shared 300-call budget TuRBO\mbox{-}1 is slightly ahead ($940$ vs.\ $933$ for ZOAF), but at $\sim 10^{3}$\,s runtime compared with 0.7\,s for ZOAF and 0.8\,s for CMA\mbox{-}ES.
With a modest extension to 396 calls, ZOAF reaches $978$ in 0.9\,s, the best power efficiency among all methods.
Overall, ZOAF offers the best accuracy–efficiency balance: it ties the top gain, strongly dominates GBP, and remains competitive or better on power efficiency while converging substantially faster. Additional statistical-significance and uncertainty results (e.g., mean/std across independent runs for gain, GBP, and power efficiency) are provided in Appendix Table~\ref{tab:gain_gbp_power}.

% \vspace{-8pt}

% ===================== 4.3 TWO-STAGE (22-PARAM) =====================
\subsection{Two-Stage Cascaded Signal-Conditioning Amplifier}

The circuit (Fig.~\ref{fig:opamp3}(b)) comprises two cascaded voltage-gain stages with feedback and compensation networks that jointly shape in-band flatness and high-frequency roll-off.
The design vector contains $22$ tunable components (stage feedbacks, compensation capacitors, interstage filter elements, and bias settings) under box bounds.
We consider \emph{three} single-objective metrics over a target passband $\mathcal{B}=[\omega_L,\omega_H]$. Let $H(j\omega)$ be the small-signal transfer function and $A_{\mathrm{DC}}(x)=|H(0)|$.

\noindent\textbf{Optimization setup.} Let $L(\omega)\triangleq 20\log_{10}|H(j\omega)|$ denote the magnitude response in dB; recall $A_{\mathrm{DC}}(x)=|H(0)|$. We consider three single-objective FOMs---peaking, ripple, and overshoot---all minimized over $x$:
\begingroup
\setlength{\abovedisplayskip}{2pt}
\setlength{\belowdisplayskip}{4pt}
\setlength{\abovedisplayshortskip}{2pt}
\setlength{\belowdisplayshortskip}{4pt}
\begin{subequations}\label{eq:obj_metrics}
\begin{align}
\mathrm{Peak}(x)      &= \max_{\omega\in\mathcal{B}}\bigl[L(\omega)-L(0)\bigr], \label{eq:obj_peak}\\
\mathrm{Ripple}(x)    &= \max_{\omega\in\mathcal{B}} L(\omega) - \min_{\omega\in\mathcal{B}} L(\omega), \label{eq:obj_ripple}\\
r(x)                  &= \frac{\max_{\omega\in\mathcal{B}}|H(j\omega)|}{|H(j2\pi f_0)|}, \label{eq:obj_os_r}\\
\zeta(x)              &= \sqrt{\bigl(1-\sqrt{1-r(x)^{-2}}\bigr)/2}, \label{eq:obj_os_zeta}\\
\mathrm{Overshoot}(x) &= 100\,\exp\!\left(\!-\frac{\pi\,\zeta(x)}{\sqrt{1-\zeta(x)^2}}\!\right)\mathbf{1}_{\{r(x)>1\}}. \label{eq:obj_os}
\end{align}
\end{subequations}
\endgroup

\begin{figure*}[!t]
\centering
\begin{minipage}[t]{0.49\textwidth}
\centering
\includegraphics[width=\textwidth]{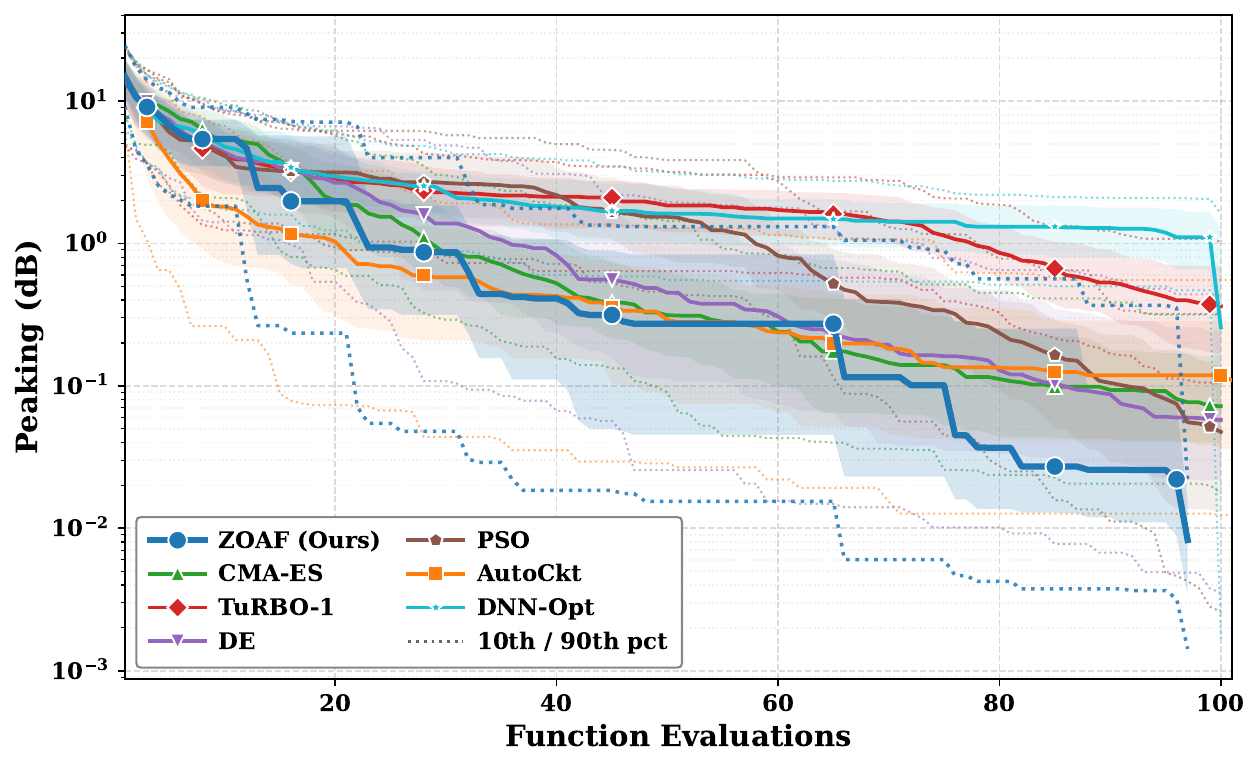}
% \vspace{1pt}
{\footnotesize\rmfamily (a) peaking}
\end{minipage}\hfill
\begin{minipage}[t]{0.49\textwidth}
\centering
\includegraphics[width=\textwidth]{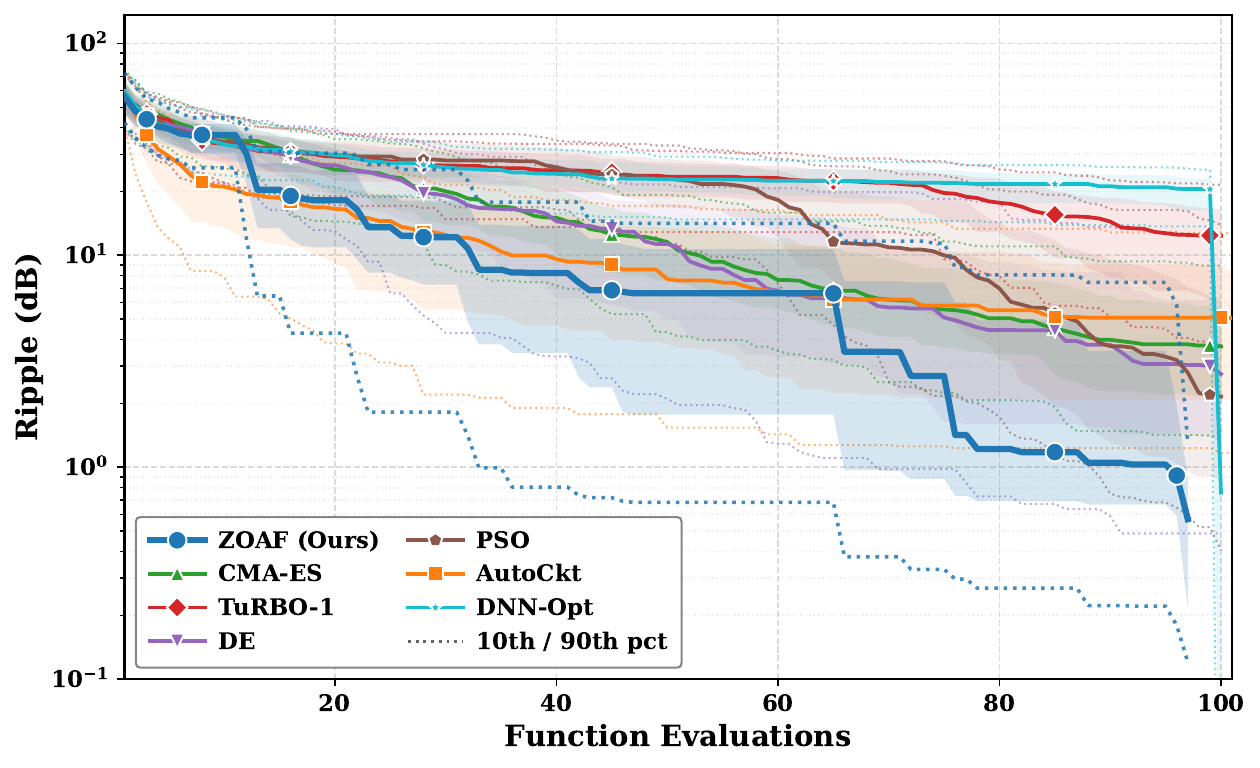}
% \vspace{1pt}
{\footnotesize\rmfamily (b) ripple}
\end{minipage}
\vspace{4pt}
\caption{Best-so-far convergence on the 22-parameter two-stage cascaded amplifier under a 100 simulator-call budget, aggregated over 100 random seeds: (a) peaking and (b) ripple. Solid lines: median across seeds; shaded bands: inter-quartile range (25th--75th percentile); dotted lines: 10th and 90th percentiles. ZOAF is shown in blue with a bolder line; the $y$-axis is logarithmic; lower is better.}
\label{fig:peaking_ripple_curves}

\vspace{1pt}
\end{figure*}

% --- Two tables side-by-side under the image ---
\begin{table*}[!t]
\centering
\begin{minipage}[t]{0.49\textwidth}
\centering
\footnotesize
\setlength{\belowcaptionskip}{4pt}
\captionof{table}{Statistical performance on the Peaking FOM over 100 independent seeds (100-simulator-call budget per run, in dB; lower is better; best in bold).}
\label{tab:peaking}
\resizebox{\linewidth}{!}{%
\begin{tabular}{lcccccc}
\toprule
\textbf{FOM: Peaking} & \textbf{Min} & \textbf{10th pct} & \textbf{Median} & \textbf{Mean} & \textbf{90th pct} & \textbf{Max} \\
\midrule
\textbf{ZOAF}    & \textbf{0.0000}$^{\dagger}$ & \textbf{0.0014} & \textbf{0.0083} & \textbf{0.0107} & \textbf{0.0214} & \textbf{0.0616} \\
PSO              & 0.0000$^{\ddagger}$         & 0.0026          & 0.0475          & 0.1314          & 0.3196          & 1.4021 \\
DE               & $3.115\!\times\!10^{-8}$    & 0.0035          & 0.0575          & 0.1660          & 0.4295          & 1.2765 \\
CMA-ES           & 0.0009                      & 0.0178          & 0.0719          & 0.1375          & 0.3163          & 0.9768 \\
AutoCkt          & $4.560\!\times\!10^{-6}$    & 0.0124          & 0.1094          & 0.2173          & 0.5512          & 1.8744 \\
DNN-Opt          & $3.087\!\times\!10^{-6}$    & 0.0017          & 0.2579          & 0.5634          & 1.6125          & 1.9566 \\
TuRBO-1          & 0.0003                      & 0.1054          & 0.3616          & 0.4580          & 1.0237          & 1.4990 \\
\bottomrule
\end{tabular}
}
\\[2pt]
{\scriptsize $^{\dagger}$One ZOAF seed hit the metric floor (floating-point $0.0$); next-smallest is $5.881\times 10^{-5}$ dB. $^{\ddagger}$Same for one PSO seed; next-smallest is $1.375\times 10^{-4}$ dB.}
\end{minipage}
\hfill
\begin{minipage}[t]{0.49\textwidth}
\centering
\footnotesize
\setlength{\belowcaptionskip}{4pt}
\captionof{table}{Statistical performance on the Ripple FOM over 100 independent seeds (100-simulator-call budget per run, in dB; lower is better; best in bold).}
\label{tab:ripple}
\resizebox{\linewidth}{!}{%
\begin{tabular}{lcccccc}
\toprule
\textbf{FOM: Ripple} & \textbf{Min} & \textbf{10th pct} & \textbf{Median} & \textbf{Mean} & \textbf{90th pct} & \textbf{Max} \\
\midrule
\textbf{ZOAF}    & 0.019          & 0.120          & \textbf{0.569} & 0.731          & \textbf{1.305} & 6.698  \\
DNN-Opt          & \textbf{0.0001}& \textbf{0.0007}& 0.759          & \textbf{0.569} & 1.457          & \textbf{3.884} \\
PSO              & 0.059          & 0.405          & 2.154          & 4.416          & 13.924         & 21.946 \\
DE               & 0.042          & 0.486          & 2.753          & 4.686          & 12.462         & 22.420 \\
CMA-ES           & 0.623          & 1.396          & 3.719          & 4.428          & 8.941          & 13.795 \\
AutoCkt          & 0.171          & 1.167          & 5.044          & 5.896          & 12.713         & 20.768 \\
TuRBO-1          & 1.269          & 3.761          & 12.369         & 11.874         & 21.524         & 27.886 \\
\bottomrule
\end{tabular}
}
\end{minipage}
\end{table*}
\vspace{6pt}

\begin{figure*}[!t]
\centering
\begin{minipage}[t]{0.49\textwidth}
\centering
\captionof{figure}{Best-so-far overshoot convergence on the 22-parameter amplifier.}
\label{fig:overshoot_curves}
\includegraphics[width=\textwidth]{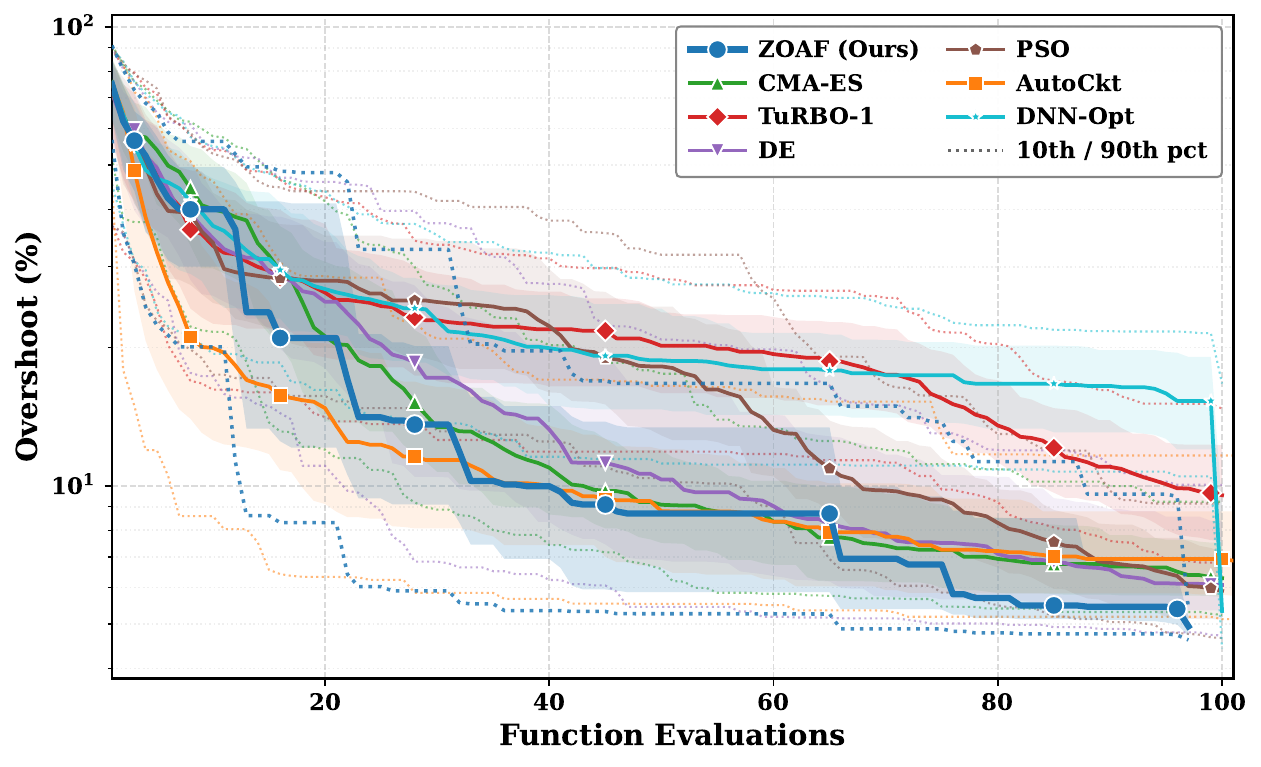}
\end{minipage}\hfill
\begin{minipage}[t]{0.49\textwidth}
\centering
\footnotesize
\captionof{table}{Statistical performance on the Overshoot FOM over 100 independent seeds (100-simulator-call budget per run, in \%; lower is better; best in bold).}
\label{tab:overshoot}
\resizebox{\linewidth}{!}{%
\begin{tabular}{lcccccc}
\toprule
\textbf{\makecell[l]{FOM:\\Overshoot (\%)}} & \textbf{Min} & \textbf{10th pct} & \textbf{Median} & \textbf{Mean} & \textbf{90th pct} & \textbf{Max} \\
\midrule
\textbf{ZOAF}    & 4.41          & \textbf{4.62} & \textbf{4.94} & \textbf{4.98} & \textbf{5.39} & \textbf{6.15} \\
DNN-Opt          & 4.35          & 4.42          & 5.34          & 7.98          & 16.70         & 19.36 \\
PSO              & 4.48          & 4.66          & 5.90          & 6.55          & 9.15          & 17.32 \\
DE               & \textbf{4.32} & 4.72          & 6.08          & 7.03          & 10.05         & 16.49 \\
CMA-ES           & 4.52          & 5.23          & 6.29          & 7.01          & 9.23          & 15.15 \\
AutoCkt          & 4.34          & 5.12          & 6.87          & 7.72          & 11.64         & 21.18 \\
TuRBO-1          & 4.44          & 6.80          & 9.54          & 10.07         & 14.73         & 17.97 \\
\bottomrule
\end{tabular}
}
\end{minipage}
\end{figure*}

\vspace{-8pt}

\noindent\textbf{Result comparison.}
Under a $\sim$100 simulator-call budget aggregated over 100 independent random seeds (Fig.~\ref{fig:peaking_ripple_curves}, Tables~\ref{tab:peaking} and \ref{tab:ripple}), ZOAF achieves the best median final value on both objectives. On peaking, ZOAF dominates every reported statistic with a median of $0.0083$\,dB---more than $5.7\times$ lower than every baseline (PSO $0.0475$, DE $0.0575$, CMA-ES $0.0719$, AutoCkt $0.1094$, DNN-Opt $0.2579$, TuRBO-1 $0.3616$\,dB). On ripple, ZOAF reaches a median of $0.569$\,dB; DNN-Opt is the closest competitor at $0.759$\,dB, and the remaining baselines range from $2.154$\,dB (PSO) to $12.369$\,dB (TuRBO-1). ZOAF also attains the lowest 90th-percentile values on both objectives ($1.305$\,dB ripple, $0.0214$\,dB peaking), indicating robustness to adversarial seeds; in contrast, DNN-Opt exhibits bimodal seed behavior on ripple, with strong best-case percentiles (min $0.0001$\,dB, 10th $0.0007$\,dB) but a median above ZOAF's, reflecting its training-then-verify schedule whose median improvement occurs only at the final SPICE call. In terms of convergence speed, ZOAF reaches TuRBO-1's 100-evaluation median in only 26 evaluations on ripple and 42 on peaking ($3.8\times$ and $2.4\times$ speedup, respectively), and matches every progressive baseline's 100-evaluation median in 26--76 simulator calls ($1.3$--$3.8\times$ speedup across baselines).
As an ablation on the hybrid schedule, for peaking the ZO-RGE phase alone brings the median best-so-far to $0.273$\,dB by eval 51 and the subsequent ZO-CGE phase further refines this to $0.0083$\,dB by eval 97; for ripple, ZO-RGE reaches $6.613$\,dB by eval 51 and ZO-CGE improves it to $0.569$\,dB by eval 97.

Under the same budget (Fig.~\ref{fig:peaking_ripple_curves}, Tables~\ref{tab:peaking} and \ref{tab:ripple}), the learning-based baselines remain less competitive than ZOAF in solution quality. Specifically, on peaking, AutoCkt and DNN-Opt reach medians of $0.1094$ and $0.2579$\,dB respectively, both more than an order of magnitude above ZOAF's $0.0083$\,dB. On ripple, AutoCkt's median is $5.044$\,dB; DNN-Opt's $0.759$\,dB is the closest competitor to ZOAF's $0.569$\,dB, but its convergence trace is effectively a step function (random sampling for 99 calls and a single Actor-recommended design at call 100), offering no usable intermediate results within the budget.

Beyond ripple and peaking, we additionally evaluate the \emph{overshoot} FOM under the same 100-simulator-call budget aggregated over 100 random seeds (Fig.~\ref{fig:overshoot_curves}, Table~\ref{tab:overshoot}). ZOAF wins five of six statistical columns---median ($4.94\%$), mean ($4.98\%$), 10th percentile ($4.62\%$), 90th percentile ($5.39\%$), and max ($6.15\%$)---and is within $0.08$ percentage points of DE's leading minimum ($4.32\%$ vs.\ ZOAF $4.41\%$). The advantage is most pronounced in the tails: every baseline's max lies between $15.15\%$ and $21.18\%$, i.e.\ $2.5$--$3.4\times$ above ZOAF's $6.15\%$, and ZOAF's mean ($4.98\%$) is $1.6\times$ lower than the next-best method (DNN-Opt at $7.98\%$) because every baseline produces a long right-tail of failed-design seeds that ZOAF's progressive refinement avoids. ZOAF is also the only method whose median crosses sub-$5\%$ within the budget (reached at evaluation 97); DNN-Opt ends at $5.34\%$ and every other baseline above $5.9\%$. On convergence speed, ZOAF matches every progressive baseline's 100-evaluation median overshoot in 42--76 simulator calls ($1.3$--$2.4\times$ speedup), with the strongest gain again vs.\ TuRBO-1 (42 evaluations, $2.4\times$). Importantly, the 100-simulator-call budget used throughout this section is not a constrained short-horizon setting but a saturated regime for this benchmark: every method's median best-so-far curve enters a flat tail well before evaluation 100 (Figs.~\ref{fig:peaking_ripple_curves} and~\ref{fig:overshoot_curves}), so further increasing the budget yields negligible additional improvement and the comparisons in Tables~\ref{tab:peaking}--\ref{tab:overshoot} reflect each method's converged performance under matched conditions.

\vspace{-4pt}

\section{Conclusion}\label{sec:conclusion}
Analog/RF sizing suffers from tight evaluation budgets and inaccessible gradients.
Model-based Bayesian optimization often relies on costly, fragile surrogates, while evolutionary algorithms tend to be evaluation-hungry and struggle with constraints.
In response, we proposed \emph{ZOAF}, a two-phase, multi-start projected ZO scheduler that uses random-direction steps for coarse search and coordinate-wise refinements for late-stage accuracy, governed by a lightweight sliding-window controller and rank-based restarts.
Across an RF matcher and two amplifiers under equal simulator-call budgets and shared initialization, ZOAF delivers faster early progress and stronger final designs: it achieves the best median final value on every reported figure of merit (with up to an order-of-magnitude advantage in median peaking on the 22-parameter two-stage amplifier), the most robust worst-case behavior across seeds (overshoot maximum $2.5$--$3.4\times$ lower than every baseline), and $1.3$--$3.8\times$ fewer simulator calls to convergence---while remaining a pure black-box, budget-aware optimizer that does not require simulator source-code access or heavy surrogate tuning.
Future work includes making ZOAF variation-aware via stochastic or variance-reduced ZO estimators, principled corner/Monte Carlo scheduling and multi-fidelity allocation, and extensions to discrete choices and multi-objective design.

\section*{Acknowledgments}
The authors thank the National Institute of Standards and Technology (NIST) for supporting this work under Award \#70NANB24H084.

\textit{Generative AI Use Disclosure.} ChatGPT and Claude were used as writing assistants for minor grammar and phrasing polish on author-drafted text. All technical content and claims are the authors' own.

\appendix
\section*{Smoothed Objective for ZO-RGE and Proofs}\label{app:zo_rge_proofs}
We focus on the (unclipped) ZO-RGE estimator in Eq.~(4) and analyze its
properties through a smoothed objective. The results are standard in zeroth-order
optimization and included here for completeness.

\section*{Setup and assumptions}
Let $u\sim \mathcal{N}(0,I_d)$ and $\mu>0$.
Define the Gaussian-smoothed objective
\begin{equation}
f_\mu(x)\triangleq \mathbb{E}_u\big[f(x+\mu u)\big].
\end{equation}
We assume:

\textbf{Assumption 1 (Integrability).} $\mathbb{E}|f(x+\mu u)|<\infty$ and
$\mathbb{E}\|\nabla f(x+\mu u)\|<\infty$ for the $x$ of interest.

\textbf{Assumption 2 ($L$-smoothness).} $f$ is differentiable and $\nabla f$ is $L$-Lipschitz:
$\|\nabla f(x)-\nabla f(y)\|\le L\|x-y\|$.

\textbf{Assumption 3 ($G$-Lipschitzness, local).} On a $\mu$-neighborhood of $\Omega$,
$|f(x)-f(y)|\le G\|x-y\|$.

\section*{Lemma: Gaussian smoothing gradient identity}
\begin{lemma}\label{lem:stein}
For $f_\mu(x)=\mathbb{E}[f(x+\mu u)]$ with $u\sim\mathcal{N}(0,I_d)$,
\begin{equation}\label{eq:grad_smooth_identity}
\nabla f_\mu(x)=\frac{1}{\mu}\mathbb{E}\big[f(x+\mu u)\,u\big].
\end{equation}
\end{lemma}
\begin{IEEEproof}
By A1 and differentiating under the expectation,
$\nabla f_\mu(x)=\mathbb{E}[\nabla f(x+\mu u)]$.
For a standard normal $u$, Stein's identity gives
$\mathbb{E}[u\,\phi(u)]=\mathbb{E}[\nabla_u \phi(u)]$ for suitable $\phi$.
Let $\phi(u)=f(x+\mu u)$; then $\nabla_u \phi(u)=\mu \nabla f(x+\mu u)$.
Thus $\mathbb{E}[u f(x+\mu u)]=\mu \mathbb{E}[\nabla f(x+\mu u)]$,
which yields~\eqref{eq:grad_smooth_identity}.
\end{IEEEproof}

\section*{Unbiasedness of the two-point ZO-RGE estimator}
Define the one-sample two-point estimator
\begin{equation}\label{eq:one_sample_est}
g(x;u)\triangleq \frac{f(x+\mu u)-f(x-\mu u)}{2\mu}\,u,
\end{equation}
and the $N$-sample averaged ZO-RGE estimator
\begin{equation}\label{eq:rge_est}
\hat g(x)\triangleq \frac{1}{N}\sum_{k=1}^N g(x;u_k),
\quad u_k \stackrel{i.i.d.}{\sim}\mathcal{N}(0,I_d).
\end{equation}

\begin{proposition}\label{prop:unbiased}
$\mathbb{E}[\hat g(x)]=\nabla f_\mu(x)$.
\end{proposition}
\begin{IEEEproof}
By i.i.d. sampling it suffices to show $\mathbb{E}[g(x;u)]=\nabla f_\mu(x)$.
Using symmetry of $u$ (i.e., $u\overset{d}{=}-u$),
\[
\mathbb{E}[f(x-\mu u)\,u]
=\mathbb{E}[f(x+\mu u)\,(-u)]
=-\mathbb{E}[f(x+\mu u)\,u].
\]
Therefore,
\begin{align}
\mathbb{E}[g(x;u)]
&= \frac{1}{2\mu}\Big(\mathbb{E}[f(x+\mu u)\,u]-\mathbb{E}[f(x-\mu u)\,u]\Big) \nonumber\\
&= \frac{1}{\mu}\mathbb{E}[f(x+\mu u)\,u]. \label{eq:appendix_eg_identity}
\end{align}
Applying Lemma~\ref{lem:stein} yields
$\mathbb{E}[g(x;u)]=\nabla f_\mu(x)$, hence $\mathbb{E}[\hat g(x)]=\nabla f_\mu(x)$.
\end{IEEEproof}

\section*{Smoothing bias bound}
\begin{proposition}\label{prop:smooth_bias}
If $f$ is $L$-smooth (Assumption~2), then
$\|\nabla f_\mu(x)-\nabla f(x)\|\le L\mu\,\mathbb{E}\|u\|\le L\mu\sqrt{d}$.
\end{proposition}
\begin{IEEEproof}
By definition and Jensen's inequality,
\begin{align}
\|\nabla f_\mu(x)-\nabla f(x)\|
&= \Big\|\mathbb{E}\big[\nabla f(x+\mu u)-\nabla f(x)\big]\Big\| \nonumber\\
&\le \mathbb{E}\big[\|\nabla f(x+\mu u)-\nabla f(x)\|\big]. \label{eq:appendix_bias_jensen}
\end{align}
Using $L$-Lipschitzness of $\nabla f$ gives
$\|\nabla f(x+\mu u)-\nabla f(x)\|\le L\mu\|u\|$,
hence $\|\nabla f_\mu(x)-\nabla f(x)\|\le L\mu\,\mathbb{E}\|u\|$.
For $u\sim\mathcal{N}(0,I_d)$, $\mathbb{E}\|u\|\le \sqrt{\mathbb{E}\|u\|^2}=\sqrt{d}$.
\end{IEEEproof}

\section*{Variance (mean-square error) scaling with $N$}
\begin{proposition}\label{prop:var}
If $f$ is $G$-Lipschitz on a $\mu$-neighborhood of $\Omega$ (Assumption~3), then
\begin{align}
\mathbb{E}\|\hat g(x)-\nabla f_\mu(x)\|^2
&\le \frac{1}{N}\,\mathbb{E}\|g(x;u)-\mathbb{E}g(x;u)\|^2 \nonumber\\
&\le \frac{1}{N}\,\mathbb{E}\|g(x;u)\|^2 \nonumber\\
&\le \frac{G^2\,\mathbb{E}\|u\|^4}{N} \nonumber\\
&= \frac{G^2\,d(d+2)}{N}. \label{eq:appendix_var_bound}
\end{align}
\end{proposition}
\begin{IEEEproof}
The first inequality is the standard variance reduction for averaging i.i.d.
samples. Next,
$\mathbb{E}\|X-\mathbb{E}X\|^2\le \mathbb{E}\|X\|^2$.
By $G$-Lipschitzness,
$|f(x+\mu u)-f(x-\mu u)|\le G\|2\mu u\|=2G\mu\|u\|$,
so from~\eqref{eq:one_sample_est},
\[
\|g(x;u)\|
=\Big\|\frac{f(x+\mu u)-f(x-\mu u)}{2\mu}\,u\Big\|
\le G\|u\|^2.
\]
Thus $\mathbb{E}\|g(x;u)\|^2\le G^2\mathbb{E}\|u\|^4$.
For $u\sim\mathcal{N}(0,I_d)$, $\|u\|^2\sim\chi^2_d$ and
$\mathbb{E}\|u\|^4=d(d+2)$, yielding the claim.
\end{IEEEproof}

\section*{Error Bound and Deterministic Property of ZO-CGE}

While ZO-RGE provides an unbiased estimate of the smoothed gradient, its sampling variance necessitates a transition to a stable estimator during late-stage refinement. 

\textbf{Assumption 4 (Third-order smoothness).} The objective $f$ is three times continuously differentiable, with third derivatives bounded by $M_3 > 0$, i.e., $|\nabla^3_{iii} f(x)| \le M_3$ for all $i \in \{1, \dots, d\}$.

\textbf{Proposition 4.} Under Assumption 4, the ZO-CGE estimator $\hat{g}_{\text{CGE}}(x)$ exhibits zero sampling variance, and its truncation error is bounded by $\mathcal{O}(\mu^2)$:
\begin{equation}
    |\hat{g}_{\text{CGE}}(x)[i] - \nabla_i f(x)| \le \frac{M_3}{6}\mu^2 = \mathcal{O}(\mu^2).
    \label{eq:cge_bound}
\end{equation}

\textit{Proof.} Taking the difference of the Taylor series expansions of $f(x + \mu e_i)$ and $f(x - \mu e_i)$ around $x$, the zeroth- and second-order terms cancel while the first-order terms combine to $2\mu\,\nabla_i f(x)$; dividing by $2\mu$ yields:
\begin{equation}
    \hat{g}_{\text{CGE}}(x)[i] = \nabla_i f(x) + \frac{\mu^2}{12} \Big[ \nabla_{iii}^3 f(x + \xi_1 \mu e_i) + \nabla_{iii}^3 f(x - \xi_2 \mu e_i) \Big],
\end{equation}
where $\xi_1, \xi_2 \in (0, 1)$. Applying the triangle inequality and Assumption 4 directly gives $|\hat{g}_{\text{CGE}}(x)[i] - \nabla_i f(x)| \le \frac{M_3}{6} \mu^2$. Since $e_i$ are fixed bases and the evaluations are deterministic, the sampling variance is exactly zero. \hfill $\blacksquare$

\textit{Remark.} As the search approaches a local optimum ($\|\nabla f(x)\| \to 0$), the $\mathcal{O}(d^2/N)$ variance of ZO-RGE would overpower the gradient signal. ZO-CGE circumvents this, providing the deterministic, coordinate-wise precision required for stringent analog design closure.

\section*{Convergence Guarantee for Non-Convex Objectives}
To guarantee that the ZO algorithm does not diverge in the highly nonconvex analog/RF search landscape, we analyze the convergence rate of the ZO-RGE phase. We consider the equivalent minimization problem $F(x) = -f(x)$ with a fixed learning rate $\eta$.

\textbf{Assumption 5 (Bounded objective).} $F(x)$ is bounded below by a finite $F^*$.

\textbf{Proposition 5.} Under $L$-smoothness (Assumption~2) and a stochastic gradient variance bounded by $\sigma^2$, applying the ZO-RGE update $x_{t+1} = x_{t} - \eta \hat{g}_{t}$ for $T$ iterations with $\eta \le \frac{1}{4Ld}$ yields:
\begin{equation}
    \frac{1}{T} \sum_{t=1}^T \mathbb{E} \|\nabla F(x_t)\|^2 \le \frac{4[F(x_0) - F^*]}{\eta T} + 2 \eta L d \sigma^2 + 2 \mu^2 L^2 d.
    \label{eq:convergence_bound}
\end{equation}

\textit{Proof.} By the $L$-smoothness of $F(x)$, the expected descent step is bounded by:
\begin{equation}
    \mathbb{E}[F(x_{t+1})] \le F(x_t) - \eta \langle \nabla F(x_t), \mathbb{E}[\hat{g}_t] \rangle + \frac{\eta^2 L}{2} \mathbb{E}[\|\hat{g}_t\|^2].
\end{equation}
Using $\mathbb{E}[\hat{g}_t] = \nabla F_{\mu}(x_t)$ (Proposition 1) and the algebraic identity $-\langle a, b \rangle = \frac{1}{2}\|a-b\|^2 - \frac{1}{2}\|a\|^2 - \frac{1}{2}\|b\|^2$, alongside the smoothing bias bound $\|\nabla F - \nabla F_{\mu}\|^2 \le \mu^2 L^2 d$ (Proposition 2), we obtain:
\begin{equation}
    \mathbb{E}[F(x_{t+1})] \le F(x_t) - \frac{\eta}{2} \|\nabla F(x_t)\|^2 + \frac{\eta \mu^2 L^2 d}{2} + \frac{\eta^2 L}{2} \mathbb{E}[\|\hat{g}_t\|^2].
\end{equation}
Rearranging to bound $\|\nabla F(x_t)\|^2$ and telescoping the sum from $t=1$ to $T$ gives:
\begin{equation}
    \frac{1}{T} \sum_{t=1}^T \mathbb{E} \|\nabla F(x_t)\|^2 \le \frac{2[F(x_0) - F^*]}{\eta T} + \mu^2 L^2 d + \frac{\eta L}{T} \sum_{t=1}^T \mathbb{E}[\|\hat{g}_t\|^2].
\end{equation}
Applying the standard second-moment bound $\mathbb{E}[\|\hat{g}_t\|^2] \le 2d\|\nabla F(x_t)\|^2 + 2d\sigma^2$ and letting $Q = \frac{1}{T}\sum_{t=1}^T \mathbb{E}\|\nabla F(x_t)\|^2$, one obtains $(1 - 2\eta Ld)\,Q \le \frac{2[F(x_0)-F^*]}{\eta T} + \mu^2 L^2 d + 2\eta Ld\sigma^2$. Under $\eta \le \frac{1}{4Ld}$, we have $1 - 2\eta Ld \ge \frac{1}{2}$; dividing both sides by $\frac{1}{2}$ yields \eqref{eq:convergence_bound}. \hfill $\blacksquare$

\textit{Remark.} Setting $\eta = \mathcal{O}(1/\sqrt{dT})$ simplifies the bound to $\mathcal{O}(\sqrt{d/T} + \mu^2 d)$. This proves that with a sufficient simulation budget $T$ and a small smoothing radius $\mu$, the gradient norm deterministically converges to a neighborhood of zero, ensuring robust sizing across multi-modal response surfaces.

\section*{Additional Gain/GBP/Power Comparison}

\begin{table*}[t]
\centering
\begingroup
\setlength{\tabcolsep}{4pt}
\renewcommand{\arraystretch}{1.0}
\footnotesize
\setlength{\belowcaptionskip}{6pt}
\caption{Supplementary results for the Three-Stage Op-Amp experiment: performance comparison of different optimization methods in terms of gain, gain-bandwidth product (GBP), and power efficiency (best/mean/std).}
\label{tab:gain_gbp_power}
\begin{tabular*}{\textwidth}{@{\extracolsep{\fill}}lccc ccc ccc}
\toprule
\multirow{2}{*}{\textbf{Method}} 
& \multicolumn{3}{c}{\textbf{Gain}}
& \multicolumn{3}{c}{\textbf{GBP (MHz)}}
& \multicolumn{3}{c}{\textbf{Power Efficiency}} \\
\cmidrule(lr){2-4} \cmidrule(lr){5-7} \cmidrule(lr){8-10}
& \textbf{Best} & \textbf{Mean} & \textbf{Std}
& \textbf{Best} & \textbf{Mean} & \textbf{Std}
& \textbf{Best} & \textbf{Mean} & \textbf{Std} \\
\midrule
ZO       & $9.73\times10^{8}$ & $9.73\times10^{8}$ & $0$
         & $1.01\times10^{7}$ & $1.01\times10^{7}$ & $0$
         & $9.33\times10^{2}$ & $787.53$ & $92$ \\

CMA-ES   & $9.73\times10^{8}$ & $9.73\times10^{8}$ & $0$
         & $5.66\times10^{4}$ & $5.05\times10^{4}$ & $9.52\times10^{2}$
         & $9.05\times10^{2}$ & $782.13$ & $82.75$ \\

TuRBO-1  & $1.48\times10^{8}$ & $1.25\times10^{8}$ & $5.42\times10^{6}$
         & $1.38\times10^{4}$ & $1.19\times10^{4}$ & $6.78\times10^{2}$
         & $9.40\times10^{2}$ & $873.38$ & $51.52$ \\

GP-EI    & $6.29\times10^{4}$ & $3.74\times10^{4}$ & $1.78\times10^{4}$
         & $1.36\times10^{4}$ & $1.01\times10^{4}$ & $2.55\times10^{2}$
         & $8.27\times10^{2}$ & $751.07$ & $42.32$ \\

GP-PI    & $8.62\times10^{5}$ & $4.75\times10^{5}$ & $1.33\times10^{5}$
         & $1.98\times10^{2}$ & $1.17\times10^{2}$ & $5.40\times10^{1}$
         & $8.04\times10^{2}$ & $721.45$ & $43.75$ \\

AutoCkt  & $2.38\times10^{7}$ & $1.41\times10^{7}$ & $8.27\times10^{6}$
         & $7.68\times10^{5}$ & $4.57\times10^{5}$ & $1.13\times10^{5}$
         & $8.25\times10^{2}$ & $696.97$ & $84.33$ \\

DNN-Opt  & $9.04\times10^{8}$ & $2.77\times10^{8}$ & $3.70\times10^{8}$
         & $5.42\times10^{6}$ & $2.84\times10^{6}$ & $2.11\times10^{6}$
         & $6.84\times10^{2}$ & $547.24$ & $75.41$ \\
\bottomrule
\end{tabular*}
\endgroup
\end{table*}

As a supplementary analysis for the Three-Stage Op-Amp experiment, Table~\ref{tab:gain_gbp_power} reports best/mean/std statistics to quantify statistical significance and uncertainty across methods, and highlights a clear trade-off across methods. The proposed ZO method reaches the strongest gain and GBP with zero run-to-run variation in this benchmark (best/mean identical), indicating highly stable convergence for these two objectives. CMA-ES remains competitive on gain but is substantially weaker on GBP, while GP-EI/GP-PI show consistently lower gain and larger dispersion. TuRBO-1 attains the highest power-efficiency best/mean values, but this advantage comes with markedly lower gain and GBP than ZO. Machine learning and reinforcement learning based baselines (AutoCkt and DNN-Opt) improve over GP methods on some metrics but still trail ZO in either absolute peak performance or consistency. Overall, these results support the use of ZO as a robust high-accuracy choice when gain/GBP quality and repeatability are prioritized, with TuRBO-1 offering an alternative when power efficiency is the primary target.

\bibliographystyle{IEEEtran}
\bibliography{refs}

\newpage

% \section*{Biography Section}
% If you have an EPS/PDF photo (graphicx package needed), extra braces are
%  needed around the contents of the optional argument to biography to prevent
%  the LaTeX parser from getting confused when it sees the complicated
%  $\backslash${\tt{includegraphics}} command within an optional argument. (You can create
%  your own custom macro containing the $\backslash${\tt{includegraphics}} command to make things
%  simpler here.)
 
% \vspace{11pt}

% \bf{If you include a photo:}\vspace{-33pt}
% \begin{IEEEbiography}[{\includegraphics[width=1in,height=1.25in,clip,keepaspectratio]{fig1}}]{Michael Shell}
% Use $\backslash${\tt{begin\{IEEEbiography\}}} and then for the 1st argument use $\backslash${\tt{includegraphics}} to declare and link the author photo.
% Use the author name as the 3rd argument followed by the biography text.
% \end{IEEEbiography}

% \vspace{11pt}

% \bf{If you will not include a photo:}\vspace{-33pt}
% \begin{IEEEbiographynophoto}{John Doe}
% Use $\backslash${\tt{begin\{IEEEbiographynophoto\}}} and the author name as the argument followed by the biography text.
% \end{IEEEbiographynophoto}

\vfill

\end{document}